\documentclass[10pt,conference]{IEEEtran}

\usepackage{epsf}
\usepackage{graphicx}
\usepackage{amsmath} \usepackage{amssymb}
\usepackage{booktabs}
\usepackage{subfigure}
\usepackage{cases}

\newenvironment{proof}{\begin{IEEEproof}}{\end{IEEEproof}}

\newtheorem{theorem}{Theorem}[section]

\newtheorem{proposition}{Proposition}[section]
\newtheorem{lemma}{Lemma}[section]
\newtheorem{corollary}{Corollary}[section]

\long\def\symbolfootnote[#1]#2{\begingroup%
\def\thefootnote{\fnsymbol{footnote}}\footnote[#1]{#2}\endgroup}

\def\dref#1{(\ref{#1})}

\def\be{\begin{equation}} \def\ee{\end{equation}}
\def\ba{\begin{array}} \def\ea{\end{array}} \def\bna{\begin{eqnarray}}
\def\ena{\end{eqnarray}}

 \def\bna{\begin{eqnarray}}
\def\ena{\end{eqnarray}} \def\dref#1{(\ref{#1})}
\begin{document}

\title{Cut-Set Bound Is Loose for Gaussian Relay Networks}
\author{\authorblockN{Xiugang Wu and Ayfer \"{O}zg\"{u}r}
\authorblockA{
Department of Electrical Engineering\\
Stanford University, Stanford, CA 94305\\
Email: x23wu@stanford.edu; aozgur@stanford.edu} }

%

\maketitle

\begin{abstract}
The cut-set bound developed by Cover and El Gamal in 1979 has since remained the best known upper bound on the capacity of the Gaussian relay channel. We develop a new upper bound on the capacity of the Gaussian primitive relay channel which is tighter than the cut-set bound. Our proof is based on typicality arguments and concentration of Gaussian measure. Combined with a simple tensorization argument proposed by Courtade and Ozgur in 2015, our result also implies that  the current capacity approximations for Gaussian relay networks, which have linear gap to the cut-set bound in the number of nodes, are order-optimal and leads to a lower bound on the pre-constant.
\end{abstract}

\section{Introduction}\label{S:Introduction}
The relay channel, and its Gaussian version in particular, models the communication scenario where a wireless link is assisted by a single relay. Motivated by the need to increase the spectral efficiency of wireless systems, characterizing the capacity of the Gaussian relay channel has been one of the central problems in information theory over the past couple of decades.

The single relay channel has been introduced by van der Meulen in \cite{van71} and the seminal work of Cover and El Gamal  in 1979 \cite{covelg79} has developed two basic achievability schemes for this setup, namely decode-and-forward and compress-and-forward, as well as an upper bound on its capacity, now known as the cut-set bound. Over the following 35 years, many new relaying strategies have been discovered such as amplify-and-forward, hash-and-forward, quantize-map-and-forward, compute-and-forward \cite{kramer,schein,KimAllerton,Avestimehretal,bobak} etc., however the cut-set bound has remained as the only upper bound on the capacity of the Gaussian relay channel. 
To our knowledge,  it is not even known if the cut-set bound is tight or not for this channel.

In this paper, we make progress on this problem by developing a new upper bound on the capacity of the Gaussian primitive relay channel.\footnote{For the sake of simplicity, in this paper we only focus on the symmetric case where the channels from the source to the relay and the destination have the same SNRs. Our arguments can be extended to the asymmetric case via channel simulation arguments.} This is a special case of the Gaussian single relay channel where the multiple access channel from the source and the relay to the destination has orthogonal components \cite{KimAllerton}. See Figure~\ref{F:GaussianRelay}. Here, the relay can be thought of as communicating to the destination over a Gaussian channel in a separate frequency band, or equivalently the destination can be thought of as equipped with two receive antennas, one directed to the source and one directed to the relay with no interference in between.\footnote{Note that due to  network equivalence, the rate limited channel from the relay to the destination in Figure~\ref{F:GaussianRelay} can be equivalently thought of  as a Gaussian channel of the same capacity \cite{Koetteretal}.} Our upper bound is tighter than the cut-set bound for this channel for all channel parameters. While this result is developed in the single-relay setting, it  has implications also for Gaussian networks with multiple relays. In particular, combined with a simple tensorization argument recently proposed in \cite{CourtadeOzgur}, it implies that the linear (in the number of nodes) gap to the cut-set bound in  current capacity approximations for Gaussian relay networks is fundamental. Indeed, the true capacity of Gaussian relay networks can have linear gap to the cut-set bound and our result can be used to obtain a lower bound on the pre-constant.

\begin{figure}[t!]
\centering
\includegraphics[width=0.3\textwidth]{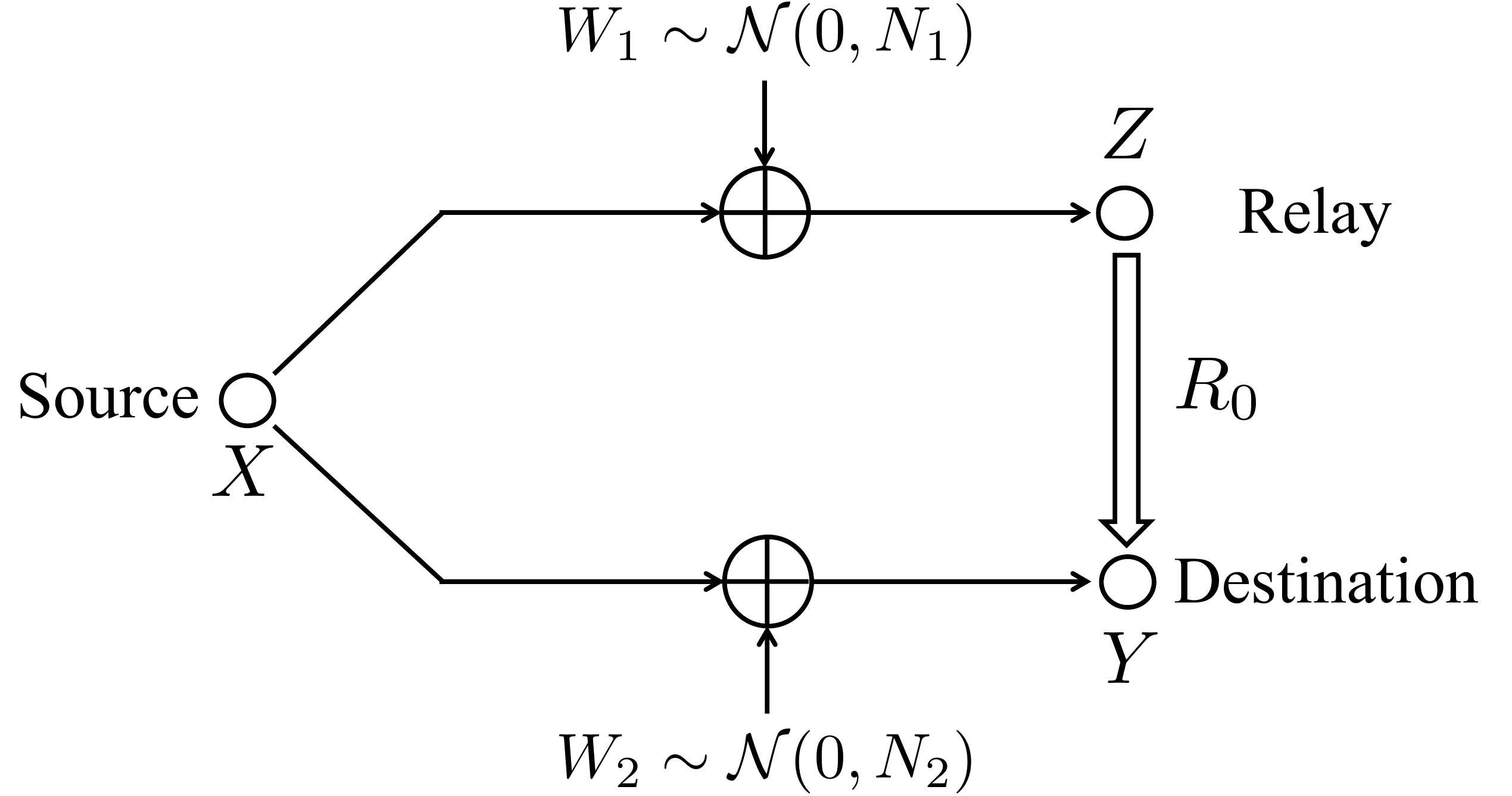}
\caption{Gaussian primitive relay channel.}
\label{F:GaussianRelay}
\end{figure}

Our upper bound builds on the approach we developed in our recent work \cite{WuXieOzgur_ISIT2015} for bounding the capacity of the discrete memoryless primitive relay channel. Similar to earlier bounds on the capacity of the discrete memoryless primitive relay channel \cite{Zhang, Xue}, the bound we developed in \cite{WuXieOzgur_ISIT2015} builds on the (generalized) blowing-up lemma, however unlike these earlier bounds does not critically rely on the finiteness of the alphabet size, which allows us to extend it to the Gaussian case in the current paper. Analogous to the results  for the discrete memoryless case \cite{Zhang, Xue,WuXieOzgur_ISIT2015}, a key ingredient of our upper bound for the Gaussian case is a Gaussian measure concentration result.


%
%
%
%

\section{Preliminaries}\label{S:ChannelModel}

\subsection{Channel Model}
Consider a Gaussian primitive relay channel as depicted in Fig. \ref{F:GaussianRelay}, where $X\in\mathbb{R}$ denotes the source signal which is constrained to average power $P$, and $Z\in\mathbb{R}$ and $Y\in\mathbb{R}$ denote the received signals of the relay and the destination. We have
\begin{numcases}{}
Z=X+W_1\nonumber \\
Y=X+W_2 \nonumber
\end{numcases}
where $W_1$ and $W_2$ are Gaussian noises that are independent of each other and $X$, and have zero mean and variances $N_1$ and $N_2$ respectively. The relay can communicate to the destination via an error-free digital link of rate $R_0$.

For this channel, a code of rate $R$ and blocklength $n$, denoted by $$(\mathcal{C}_{(n,R)}, f_n(z^n), g_n(y^n,f_n(z^n))), \mbox{ or simply, } (\mathcal{C}_{(n,R)}, f_n, g_n), $$
consists of the following:
\begin{enumerate}
  \item A codebook at the source $X$,
$$\mathcal{C}_{(n,R)}=\{x^n (m), m\in \{1,2,\ldots, 2^{nR}\} \}$$
where $$\frac{1}{n}\sum_{i=1}^n  x_i^2 (m)\leq P, \ \forall m \in \{1,2,\ldots,2^{nR}\};$$
  \item An encoding function at the relay $Z$,
$$f_n: \mathbb R^n \rightarrow \{1,2,\ldots, 2^{nR_0}\};$$
  \item A decoding function at the destination $Y$,
$$g_n: \mathbb R^n  \times \{1,2,\ldots, 2^{nR_0}\}  \rightarrow \{1,2,\ldots, 2^{nR}\}.$$
\end{enumerate}

The average probability of error of the code is defined as
$$P_e^{(n)}=\mbox{Pr}(g_n(Y^n,f_n(Z^n)) \neq M ),$$
where the message $M$ is assumed to be uniformly drawn from the message set $ \{1,2,\ldots, 2^{nR}\}$. A rate $R$ is said to be achievable if there exists a sequence of codes
$$\{(\mathcal{C}_{(n,R)}, f_n, g_n)\}_{n=1}^{\infty}$$
such that the average probability of error $P_e^{(n)} \to 0$ as $n \to \infty$.
The capacity of the primitive relay channel is the supremum of all achievable rates, denoted by $C(R_0)$.

\subsection{The Cut-Set Bound}

For the Gaussian primitive relay channel, the cut-set bound can be stated as follows.

\begin{proposition}[Cut-set Bound]\label{P:cutset}
For the Gaussian primitive relay channel, if a rate $R$ is achievable, then there exists a random variable $X$ satisfying $E[X^2]\leq P$ such that
\begin{numcases}{}
 R   \leq I(X;Y,Z)\label{E: cut1} \\
R    \leq   I(X;Y)+R_0  \label{E: cut2}.
\end{numcases}
\end{proposition}

It can be easily shown that both $I(X;Y,Z)$ and $I(X;Y)$ in Proposition \ref{P:cutset} are maximized when $X\sim\mathcal N (0,P)$, leading us to the following corollary.

\begin{corollary}\label{C:cutset}
For the  Gaussian primitive relay channel, if a rate $R$ is achievable, then
\begin{numcases}{}
 R   \leq \frac{1}{2}\log \left(1+\frac{P}{N_1}+\frac{P}{N_2}\right)  \\
R    \leq   \frac{1}{2}\log \left(1+\frac{P}{N_2}\right)+R_0  .
\end{numcases}
\end{corollary}



\section{Main Result}\label{S:mainresults}

To simplify the exposition, in this paper we only concentrate on the symmetric case of the Gaussian primitive relay channel, that is, when $N_1=N_2=:N$. Our results can be extended to the asymmetric case by using channel simulation arguments. We defer this extension to the longer version of the paper. The following theorem states the main result of this paper.

\begin{theorem} \label{T:newboundsym}
For the symmetric Gaussian primitive relay channel, if a rate $R$ is achievable, then there exists a random variable $X$ satisfying $E[X^2]\leq P$ and some $a\in [0,R_0]$ such that
\begin{numcases}{}
 R   \leq I(X;Y,Z)\label{E: symnew1} \\
R    \leq   I(X;Y)+R_0  -a \label{E: symnew2}\\
R    \leq   I(X;Y)+a+\sqrt{2a\ln 2}\log e . \label{E: symnew3}
\end{numcases}
\end{theorem}

As in the case of the cut-set bound, since both $I(X;Y,Z)$ and $I(X;Y)$ in Theorem \ref{T:newboundsym} are maximized when $X\sim\mathcal N (0,P)$, we have the following corollary.
\begin{corollary}\label{C:newboundsym}
For the symmetric Gaussian primitive relay channel, if a rate $R$ is achievable, then there exists some $a\in [0,R_0]$ such that
\begin{numcases}{}
 R   \leq \frac{1}{2}\log \left(1+\frac{2P}{N}\right)\label{E:newboundsym1}  \\
R    \leq   \frac{1}{2}\log \left(1+\frac{P}{N}\right)+R_0  -a  \label{E:newboundsym2} \\
R    \leq   \frac{1}{2}\log \left(1+\frac{P}{N}\right)+a+\sqrt{2a\ln2}\log e. \label{E:newboundsym3}
\end{numcases}
\end{corollary}

Note that in the symmetric case, by Corollary \ref{C:cutset}, the cut-set bound says that if a rate $R$ is achievable, then
\begin{numcases}{}
 R   \leq \frac{1}{2}\log \left(1+\frac{2P}{N}\right)  \label{E:cutsetsym1} \\
R    \leq   \frac{1}{2}\log \left(1+\frac{P}{N}\right)+R_0.\label{E:cutsetsym2}
\end{numcases}

Clearly the bound on $R$ in Corollary~\ref{C:newboundsym} is tighter than the cut-set bound since \eqref{E:newboundsym2} will only reduce to \eqref{E:cutsetsym2} if $a=0$. However, if $a=0$ then \eqref{E:newboundsym3} will constrain the rate $R$ by the capacity of the source-destination link. The constraint on $R$, jointly imposed by \eqref{E:newboundsym2} and \eqref{E:newboundsym3} can be found by equating them to yield
\begin{align}
R_0=2a^*+ \sqrt{2a^*\ln2} \log e. \label{E:solvegap}
\end{align}
Corollary~\ref{C:newboundsym} can be restated in terms of $a^*$ as follows: if a rate $R$ is achievable, then
\begin{numcases}{}
R   \leq \frac{1}{2}\log \left(1+\frac{2P}{N}\right)  \nonumber \\
R    \leq   \frac{1}{2}\log \left(1+\frac{P}{N}\right)+R_0-a^*. \nonumber
\end{numcases}

Note that both the cut-set bound and our new bound depend on the channel parameters through $\frac{P}{N}$ and $R_0$. It is interesting to evaluate the largest gap between these two bounds over all parameter values for the symmetric Gaussian primitive relay channel.
For this, it can be shown that when $\frac{P}{N}\rightarrow\infty$ and $R_0=0.5$, the gap takes its largest value and is given by the solution of equation \dref{E:solvegap}, which is $a^*=0.0535$.
We formally summarize this observation in the following proposition.

\begin{proposition}\label{P:largestpossible}
 Let $\Delta\left(\frac{P}{N},R_0\right)$ denote the gap between the two bounds and $\Delta^*$ its largest possible  value over all symmetric Gaussian primitive relay channels, i.e.,
$$\Delta^*:=\sup_{\frac{P}{N}, R_0} \Delta\left(\frac{P}{N},R_0\right).$$
Then, $\Delta^* =  \Delta(\infty,0.5)=0.0535$.
\end{proposition}

%
%
%
%
%

\subsection{Gaussian Relay Networks}\label{S:lineargap}

While the setup we consider in this paper can be regarded as a special case of a Gaussian relay network, the upper bound we develop for this special case can be used to infer how tightly the capacity of general Gaussian relay networks can be approximated by the cut-set bound. Initiated by the work of Avestimehr, Diggavi and Tse \cite{Avestimehretal}, there has been significant recent interest \cite{OzgurDiggavi,Limetal} in approximating the capacity of general Gaussian relay networks with the cut-set bound, i.e. bounding the gap between the rates achieved by specific schemes and the cut-set bound on capacity. The gap in these approximation results is linear in the number of nodes in the network but independent of the channel SNRs and network topology. In particular, the best currently known approximation result \cite{DDF} has a gap of $0.5N$ where $N$ is the total number of nodes. While some recent works \cite{Urs1,Urs2,Bobbie,Ritesh} demonstrate sublinear in the number of nodes (or in the total number of antennas in the case of multiple antenna nodes) gap to the cut-set bound for specific topologies, a recent tensorization argument proposed in \cite{CourtadeOzgur} shows that the gap between the capacity and the cut-set bound can be bounded by a sublinear function of the number of nodes, independent of network topology
and channel configurations, if, and only if, capacity is equal to
the cut-set bound for \emph{all} Gaussian relay networks. 
Moreover, Theorem~3 of \cite{CourtadeOzgur} implies that an explicit gap to the cut-set bound for any specific network with specific channel parameters and topology would imply a lower bound on the preconstant in these approximation results. In particular, the gap $0.0535$ in Proposition~\ref{P:largestpossible} for the Gaussian primitive relay channel (which can be thought of as a Gaussian network with two receive antennas at the destination, so four antennas in total) implies that the capacity of Gaussian relay networks can not be approximated by the cutset bound, independent of the topology and channel coefficients, with a gap that is smaller than $(0.0535/4)N\approx 0.01 N$.

\section{Proof of Theorem~\ref{T:newboundsym}}
In this section, we prove bounds \dref{E: symnew1}--\dref{E: symnew3} sequentially with the focus on showing \dref{E: symnew3}.

Suppose a rate $R$ is achievable. Then there exists a sequence of codes
\begin{align}
\{(\mathcal{C}_{(n,R)}, f_n, g_n)\}_{n=1}^{\infty}  \label{E:ReliableCodes}
\end{align}
such that the average probability of error $P_e^{(n)} \to 0$ as $n \to \infty$.

For this sequence of codes, we have
\begin{align}
n R &= H(M)\nonumber \\
&=I(M;Y^n,Z^n)+H(M|Y^n,Z^n)\nonumber \\
&\leq I(X^n;Y^n,Z^n)+H(M|Y^n,f_n(Z^n))\nonumber \\
&\leq I (X^n;Y^n,Z^n)+n \mu  \label{E:GaussianFano} \\
&= h(Y^n,Z^n)- h(Y^n,Z^n|X^n)+n  \mu   \nonumber \\
&= \sum_{i}^{n} [h(Y_i,Z_i|Y^{i-1},Z^{i-1})- h(Y_i,Z_i|X_i)] +n \mu  \nonumber \\
&\leq  \sum_{i}^{n} [h(Y_i,Z_i)- h(Y_i,Z_i|X_i)] +n \mu  \nonumber \\
&=  \sum_{i}^{n} I (X_i;Y_i,Z_i) +n \mu   \nonumber \\
&= n(I(X_Q;Y_Q,Z_Q|Q)+ \mu ) \label{E:Timesharing}\\
&= n(h(Y_Q,Z_Q|Q)- h(Y_Q,Z_Q|Q,X_Q)  +\mu ) \nonumber\\
&\leq  n(h(Y_Q,Z_Q)- h(Y_Q,Z_Q|X_Q)  + \mu ) \nonumber\\
&= n(I(X_Q;Y_Q,Z_Q)+ \mu ) \nonumber
\end{align}
i.e.,
\begin{align}
R\leq  I(X_Q;Y_Q,Z_Q)+ \mu  \label{E:summarize1}
\end{align}
for any $\mu  >0$ and sufficiently large $n$, where \dref{E:GaussianFano} follows from Fano's inequality,  \dref{E:Timesharing} follows by defining the time sharing random variable $Q$ to be uniformly distributed over $[1:n]$, and
\begin{align}\label{eq:powerconst}
E[X^2_Q]=  \frac{1}{n}\sum_{i}^{n} E[X^2_i]  =  \frac{1}{n} E\left[\sum_{i}^{n}X^2_i\right]  \leq P.
\end{align}

Moreover, for any $\mu >0$ and sufficiently large $n$,
\begin{align}
n R &=  H(M)\nonumber \\
&=I(M;Y^n,f_n(Z^n))+H(M|Y^n,f_n(Z^n))\nonumber \\
&\leq I (X^n;Y^n,f_n(Z^n))+n\mu  \label{eq:int} \\
&=I(X^n;Y^n)+I(X^n;f_n(Z^n)|Y^n)+ n\mu  \nonumber \\
&=I(X^n;Y^n)+H(f_n(Z^n)|Y^n)-H(f_n(Z^n)|X^n)+ n\mu  \nonumber \\
&\leq n(I(X_Q;Y_Q)+R_0 -a_n+ \mu ), \end{align}
i.e.,
\begin{align}
 R \leq  I(X_Q;Y_Q)+R_0 -a_n+ \mu ,  \label{E:summarize2}
\end{align}
where $a_n=\frac{1}{n}H(I_n|X^n)$ with $I_n:=f_n(Z^n)$, and $a_n$ satisfies
\begin{align}\label{E:constraint on a}
0 \leq a_n \leq R_0.
\end{align}

\begin{figure}
\centering
\includegraphics[width=0.25\textwidth]{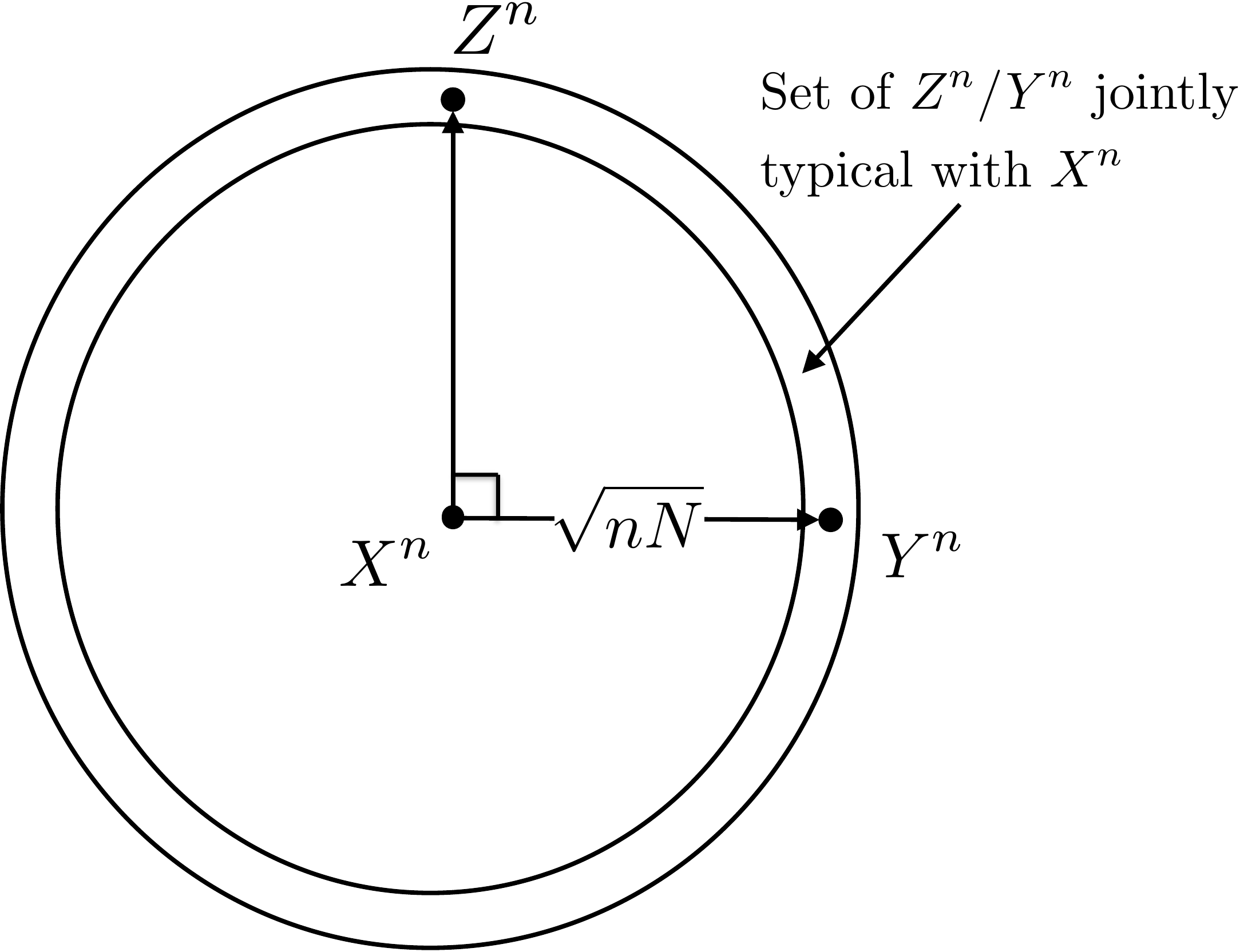}
\caption{Jointly typical set with $X^n$.}
\label{F:Typicalset}
\end{figure}

So far we have made only standard information theoretic arguments and in particular recovered the cut-set bound; note that the fact that $a_n\geq 0$ together with \eqref{E:summarize1}, \eqref{E:summarize2} and \eqref{eq:powerconst} yields the cut-set bound given in Proposition~\ref{P:cutset}. However, instead of simply lower bounding $a_n$ by $0$ in \eqref{E:summarize2}, in the sequel we will prove a third inequality involving $a_n$, which will force $a_n$ to be strictly larger than $0$. Indeed, it is intuitively easy to see that $a_n$ can not be arbitrarily small. Assume $a_n=\frac{1}{n}H(I_n|X^n)\approx 0$. Roughly speaking, this implies that given the transmitted codeword $X^n$, there is no ambiguity about $I_n$, or equivalently all $Z^n$ sequences jointly typical with $X^n$ are mapped to the same $I_n$. See Figure~\ref{F:Typicalset}. However, since $Y^n$ and $Z^n$ are statistically equivalent given $X^n$ (they share the same typical set given $X^n$) this would imply that
$I_n$ can be determined based on $Y^n$ and therefore the transmitted codeword $X^n$ can be decoded based solely on $Y^n$. This will force the rate to be smaller than $I(X_Q;Y_Q)$. In general, there is a trade-off between how close the rate can get to the multiple access bound $I(X_Q;Y_Q)+R_0$ and how much it can exceed the point-to-point capacity $I(X_Q;Y_Q)$ of the $X$-$Y$ link. We capture this trade-off as follows.

%

From \eqref{eq:int},
\begin{align}
n R &\leq I (X^n;Y^n,I_n)+n\mu \nonumber\\
&= I (X^n;I_n)+I (X^n;Y^n|I_n)+n\mu \nonumber\\
&= H (X^n)-nb_n+h(Y^n|I_n)-\frac{n}{2}\log 2\pi eN+n\mu,\label{eq:int2}
\end{align}
where we define $b_n:=\frac{1}{n}H(X^n|I_n)$ and use the fact that $h(Y^n|X^n,I_n)=h(Y^n|X^n)=\frac{n}{2}\log 2\pi eN$. In Section \ref{SS:Provingineq}, we prove the following key upper bound on the conditional entropy of $Y^n$ given the relay's transmission $I_n$,
\begin{equation}\label{eq:upperbound}
h(Y^n|I_n)\leq n(b_n-c_n + \frac{1}{2}\log 2\pi eN +a_n+\sqrt{2a_n\ln2}\log e  )
\end{equation}
where $c_n:=\frac{1}{n}H(X^n|Z^n)$. Combined with \eqref{eq:int2}, this yields
\begin{align}
n R&\leq I (X^n;Z^n)+n(a_n+\sqrt{2a_n\ln2}\log e ) +n\mu,\nonumber\\
&=I (X^n;Y^n)+n(a_n+\sqrt{2a_n\ln2}\log e ) +n\mu,\nonumber\\
&\leq n(I(X_Q;Y_Q)+a_n+\sqrt{2a_n\ln2}\log e  +\mu)\nonumber
\end{align}
where the equality follows from the fact that $Z^n$ and $Y^n$ are statistically equivalent given $X^n$. Equivalently,
\begin{align}
R \leq I(X_Q;Y_Q)+a_n+\sqrt{2a_n\ln2}\log e   +\mu.   \label{E:summarize3}\end{align}

Combining \dref{E:summarize1}, \dref{E:summarize2} and \dref{E:summarize3} , we have that if a rate $R$ is achievable, then for any $\mu >0$ and sufficiently large $n$,
\begin{numcases}{}
 R   \leq I(X_Q;Y_Q,Z_Q)+\mu  \nonumber \\
R    \leq   I(X_Q;Y_Q)+R_0  -a_n +\mu  \nonumber \\
R    \leq   I(X_Q;Y_Q)+a_n+\sqrt{2a_n\ln2}\log e   +\mu  \nonumber
\end{numcases}
where $E[X_Q^2]\leq P$ and $a_n\in [0,R_0]$. Since $\mu$ can be arbitrarily small, this proves Theorem \ref{T:newboundsym}.

\subsection{Proving Inequality \eqref{eq:upperbound} }\label{SS:Provingineq}


The remaining step then is to prove the relation in \eqref{eq:upperbound}. To prove this inequality we will look at $B$-length i.i.d. sequences of the random vectors $X^n, Y^n, Z^n,$ and $I_n$,  and derive some typicality properties for these sequences which hold with high probability when $B$ is large.\footnote{Note that $X^n$ here is a discrete random vector whose distribution is dictated by the uniform distribution on the set of possible messages and the source codebook, $Y^n$ and $Z^n$ are continuous random vectors and $I_n$ is an integer valued random variable.}

Specifically, consider the following $B$-length i.i.d. sequence
\begin{align}
\{   (X^n(b),Y^n(b),Z^n(b),I_n(b)   )    \}_{b=1}^{B}, \label{E:iidextension}
\end{align}
where for any $b\in [1:B]$, $(X^n(b),Y^n(b),Z^n(b),I_n(b))$ has the same distribution as $(X^n, Y^n, Z^n,I_n)$.  For notational convenience,  in the sequel we write the $B$-length sequence $[X^n(1),X^n(2),\ldots,X^n(B)]$ as $\mathbf X$ and similarly define $\mathbf Y, \mathbf Z$ and $\mathbf I$; note here we have
$\mathbf I=[f_n(Z^n(1)), f_n(Z^n(2)),\ldots,f_n(Z^n(B)) ]=:f(\mathbf Z)$.

We now present a key lemma in  our proof, which gives a lower bound on the conditional probability density $f(\mathbf y| \mathbf i)$ for a set of ``typical'' $(\mathbf y, \mathbf i)$ pairs. The formal proof of this lemma is delegated to  Appendix~\ref{A:ProoftoKeylemma}. In the next subsection we provide a proof sketch for this lemma that summarizes the main ideas.

\medskip
\begin{lemma}\label{L:Keylemma}
For any $\delta>0$ and sufficiently large $B$, there exists a set $\mathcal I$ of $\mathbf{i}$ such that
 $$\mbox{Pr}(\mathbf I \in \mathcal I) \geq  1-\delta,$$ and for any $\mathbf i \in \mathcal I$, there exists a set $\mathcal{Y}_\mathbf{i}$ of $\mathbf y$ such that
 $$\mbox{Pr}(\mathbf Y \in \mathcal{Y}_\mathbf{i}|\mathbf i) \geq  1-\delta,$$ and for any $\mathbf y \in \mathcal{Y}_\mathbf{i}$
\begin{align*}
f(\mathbf{y}|\mathbf{i})\geq 2^{-nB(b_n-c_n + \frac{1}{2}\log 2\pi eN +  a_n+\sqrt{2a_n\ln2}\log e    + \delta_1)}, \end{align*}
where $\delta_1 \to 0$ as $\delta \to 0$.
\end{lemma}

Equipped with this lemma, it is not difficult to prove \eqref{eq:upperbound}. For this, first consider $h(\mathbf{Y} |\mathbf{i})$ for any $\mathbf{i} \in \mathcal I$. We have
\begin{align}
h(\mathbf{Y} |\mathbf{i}) &\leq h(\mathbf{Y} |\mathbf{i})+1 - I(\mathbf{Y}; \mathbb{I}(\mathbf{Y} \in \mathcal{Y}_\mathbf{i})|\mathbf{i} ) \label{E:indicator} \\
\nonumber &=  1+h(\mathbf{Y} |\mathbb{I}(\mathbf{Y} \in \mathcal{Y}_\mathbf{i}), \mathbf{i}) \\
&=1+\mbox{Pr}(\mathbf{Y} \in \mathcal{Y}_\mathbf{i}|\mathbf{i} ) h(\mathbf{Y} |\mathbf{i},\mathbf{Y} \in \mathcal{Y}_\mathbf{i}) \nonumber  \\
&~~~~~+\mbox{Pr}(\mathbf{Y} \notin \mathcal{Y}_\mathbf{i}|\mathbf{i} ) h(\mathbf{Y} |\mathbf{i},\mathbf{Y} \notin \mathcal{Y}_\mathbf{i}),\label{E:pluggingH(Y|i)}
\end{align}
where $\mathbb I(A)$ is the indicator function defined as 1 if $A$ holds and 0 otherwise, and  \dref{E:indicator} follows since
$$I(\mathbf{Y}; \mathbb{I}(\mathbf{Y} \in \mathcal{Y}_\mathbf{i})|\mathbf{i} )\leq H(\mathbb{I}(\mathbf{Y} \in \mathcal{Y}_\mathbf{i})|\mathbf{i} )\leq 1.$$

To bound $h(\mathbf{Y} |\mathbf{i},\mathbf{Y} \in \mathcal{Y}_\mathbf{i})$, we have by Lemma \ref{L:Keylemma} that,
\begin{align}
&h(\mathbf{Y} |\mathbf{i},\mathbf{Y} \in \mathcal{Y}_\mathbf{i})\nonumber \\
& =- \int_{\mathbf y \in \mathcal{Y}_\mathbf{i}} f(\mathbf y|\mathbf{i},\mathbf{Y} \in \mathcal{Y}_\mathbf{i} )
\log f(\mathbf y|\mathbf{i},\mathbf{Y} \in \mathcal{Y}_\mathbf{i} ) d\mathbf y \nonumber \\
&\leq  - \int_{\mathbf y \in \mathcal{Y}_\mathbf{i}} f(\mathbf y|\mathbf{i},\mathbf{Y} \in \mathcal{Y}_\mathbf{i} )
\log f(\mathbf y|\mathbf{i} ) d\mathbf y \nonumber \\
&\leq nB(b_n-c_n + \frac{1}{2}\log 2\pi eN + a_n+\sqrt{2a_n\ln2}\log e    + \delta_1)  . \label{E:pluggedH(Y|i)1}\end{align}

Now consider $E [   ||\mathbf{Y}||^2   |\mathbf{i}] $ for any $\mathbf i$. We have
\begin{align*}
E [   ||\mathbf{Y}||^2   |\mathbf{i}] &=  E [   ||\mathbf{X}||^2   |\mathbf{i}] +  E [   ||\mathbf{W_2}||^2   |\mathbf{i}]\leq nB(P+N),
\end{align*}
where the equality follows from the independence between $\mathbf{X}$ and $\mathbf{W_2}$ even conditioned on $\mathbf{i}$.
Therefore,    $$E [   ||\mathbf{Y}||^2   |\mathbf{i},\mathbf{Y} \notin \mathcal{Y}_\mathbf{i}] \leq \frac{E [   ||\mathbf{Y}||^2   |\mathbf{i}] }{   \mbox{Pr}(\mathbf{Y} \notin \mathcal{Y}_\mathbf{i}|\mathbf{i} )  }\leq \frac{nB(P+N)}{   \mbox{Pr}(\mathbf{Y} \notin \mathcal{Y}_\mathbf{i}|\mathbf{i} )  }  ,$$
and
\begin{align}
&\ ~\ ~\mbox{Pr}(\mathbf{Y} \notin \mathcal{Y}_\mathbf{i}|\mathbf{i} ) h(\mathbf{Y} |\mathbf{i},\mathbf{Y} \notin \mathcal{Y}_\mathbf{i})\nonumber \\
& \leq\frac{nB}{2} \mbox{Pr}(\mathbf{Y} \notin \mathcal{Y}_\mathbf{i}|\mathbf{i} )  \log 2\pi e \frac{P+N}{   \mbox{Pr}(\mathbf{Y} \notin \mathcal{Y}_\mathbf{i}|\mathbf{i} )} \nonumber \\
& \leq  nB \delta_2 ,\label{E:pluggedH(Y|i)2}
\end{align}
for some $\delta_2 \to 0$ as $\delta \to 0$.

Plugging \dref{E:pluggedH(Y|i)1}  and \dref{E:pluggedH(Y|i)2}   into \dref{E:pluggingH(Y|i)}, we have for any $\mathbf{i} \in \mathcal I$,
\begin{align}
&h(\mathbf{Y} |\mathbf{i})\nonumber\\
&\leq 1+ \mbox{Pr}(\mathbf{Y} \in \mathcal{Y}_\mathbf{i}|\mathbf{i} ) nB[b_n-c_n +  \frac{1}{2}\log 2\pi eN \nonumber \\
 &~~~~~~~~~~~~~~~~~~~~~~~~~~~+  a_n+\sqrt{2a_n\ln2}\log e    + \delta_1] + nB\delta_2\nonumber \\
\nonumber  &=  nB(b_n-c_n + \frac{1}{2}\log 2\pi eN + a_n+\sqrt{2a_n\ln2}\log e  + \delta_3)
\end{align}
where $\delta_3 \to 0$ as $\delta \to 0$ and $B\to \infty$. Therefore, for sufficiently large $B$,
\begin{align}
&h(\mathbf{Y}|\mathbf{I})\nonumber \\
&=\sum_{i} p(\mathbf{i})h(\mathbf{Y} |\mathbf{i})\nonumber \\
&= \sum_{\mathbf{i}\in \mathcal I} p(\mathbf{i})h(\mathbf{Y} |\mathbf{i}) + \sum_{\mathbf{i} \not\in \mathcal I} p(\mathbf{i})h(\mathbf{Y}|\mathbf{i})\nonumber \\
&\leq \sum_{\mathbf{i}\in \mathcal I} p(\mathbf{i})nB(b_n-c_n + \frac{1}{2}\log 2\pi eN +  a_n  \nonumber \\
&~~~+\sqrt{2a_n\ln2}\log e+ \delta_3)+ \sum_{\mathbf{i} \not\in \mathcal I} p(\mathbf{i})  \frac{nB}{2}\log 2\pi e (P+N)       \nonumber \\
&= nB(b_n-c_n + \frac{1}{2}\log 2\pi eN + a_n+\sqrt{2a_n\ln2}\log e  + \delta_4)  \label{E:compareupper}
\end{align}
where $\delta_4 \to 0$ as $\delta \to 0$ and $B\to \infty$.
Finally observing that
\begin{align}
h(\mathbf{Y}|\mathbf{I}) = \sum_{b=1}^{B}h( Y^n(b) |  I_n(b)) =Bh(Y^n|I_n)\nonumber
\end{align}
and taking $B\to \infty$ complete the proof of inequality \eqref{eq:upperbound}.

\subsection{Proof Idea for Lemma~\ref{L:Keylemma}}

We now provide a proof sketch for Lemma~\ref{L:Keylemma}. The formal proof can be found in Appendix~\ref{A:ProoftoKeylemma}.

By the law of large numbers, if $H(I_n|X^n)=na_n$, then given a typical $(\mathbf x, \mathbf i)$ pair, it can be shown that
$$\mbox{Pr}(\mathbf{Z} \in \mathcal Z_{(\mathbf x, \mathbf i)}|\mathbf x)\doteq2^{-B\,I(Z^n;I_n|X^n)}= 2^{-nBa_n},$$
where $\mathcal Z_{(\mathbf x, \mathbf i)}$ can be roughly viewed as the set of $\mathbf z$ that are jointly typical with $(\mathbf x, \mathbf i)$.

Now we apply the following lemma, whose proof relies on a Gaussian measure concentration result and is included in Appendix \ref{A:ProoftoTalagrand}.
\begin{lemma}\label{L:Talagrand}
Let $U_1, U_2, \ldots, U_n$ be $n$ i.i.d. Gaussian random variables with $U_i \sim \mathcal{N}(0,N),\forall i \in \{1,2,\ldots,n\}$. Then, for any $A\subseteq \mathbb R^{n}$ with $\mbox{Pr}(U^n \in A)\geq 2^{-na_n}$,
\begin{align*}
\mbox{Pr}(U^n \in \Gamma_{\sqrt{n}( \sqrt{2Na_n\ln2} +r ) } (A) )\geq 1-2^{-\frac{nr^2}{2N}},\forall r>0,
\end{align*}
where
\begin{align*}
&\Gamma_{\sqrt{n}( \sqrt{2Na_n\ln2} +r ) } (A)\\
&:=\{\underline{\mathbf{\omega}}\in \mathbb R^{n}: \exists \  \underline{\mathbf{\omega}}'\in A \text{~s.t.~} d(\underline{\mathbf{\omega}},\underline{\mathbf{\omega}}') \leq \sqrt{n} (\sqrt{2Na_n\ln2} +r ) \},
\end{align*}
with $d(\underline{\mathbf{\omega}},\underline{\mathbf{\omega}}') :=||\underline{\mathbf{\omega}}-\underline{\mathbf{\omega}}' ||$ denoting the Euclidean distance between $\underline{\mathbf{\omega}}$ and $\underline{\mathbf{\omega}}'$.
\end{lemma}

With Lemma  \ref{L:Talagrand}, it can be shown that if one blows up $\mathcal Z_{(\mathbf x, \mathbf i)}$ with a radius $\sqrt{nB} \sqrt{2Na_n\ln2} $, the resultant set, denoted by $\Gamma_{\sqrt{nB} \sqrt{2Na_n\ln2} }(\mathcal Z_{(\mathbf x, \mathbf i)})$, has probability nearly 1, i.e.,
\begin{align}\label{E:blowingup}
\mbox{Pr}(\mathbf Z \in \Gamma_{\sqrt{nB} \sqrt{2Na_n\ln2}}(\mathcal Z_{(\mathbf x, \mathbf i)})|\mathbf x)\approx 1.
\end{align}
Due to the symmetry of the channel, \dref{E:blowingup} still holds with $\mathbf Z$ replaced by $\mathbf Y$. 

Now given a typical $(\mathbf x, \mathbf i)$ pair, we lower bound the conditional density $f(\mathbf y|\mathbf i)$ for all $\mathbf y\in\Gamma_{\sqrt{nB} \sqrt{2Na_n\ln2}}(\mathcal Z_{(\mathbf x, \mathbf i)})$. Given such $\mathbf y$, there exists some $\mathbf{z}\in \mathcal Z_{(\mathbf x, \mathbf i)}$ such that $ d(\mathbf y ,  \mathbf{z} )\leq \sqrt{nB} \sqrt{2Na_n\ln2}$. Consider the set of all $\mathbf x$ that are jointly typical with this $\mathbf{z}$. 
It can be shown that the $\mathbf x$'s that are jointly typical with a given $\mathbf{z}\in \mathcal Z_{(\mathbf x, \mathbf i)}$ are such that
$$
d(\mathbf x, \mathbf z)\leq \sqrt{nBN},
$$
and
$$p(\mathbf x|\mathbf i)\doteq 2^{-nBb_n}.$$
 Therefore for each $\mathbf x$ in this set
\begin{align*}
d(\mathbf x, \mathbf y)&\leq d(\mathbf x, \mathbf z)+d(\mathbf z, \mathbf y)\\
&\leq \sqrt{nB}(\sqrt{N}+\sqrt{2Na_n\ln 2}),
\end{align*}
which leads to the following lower bound on $f(\mathbf y|\mathbf x)$,
$$f(\mathbf y|\mathbf x) \stackrel {.}{\geq}  2^{ -nB \left( \frac{1}{2}\log 2\pi e N +a_n+\sqrt{2a_n\ln2}\log e       \right) },$$
by using the fact that $\mathbf y$ is Gaussian given $\mathbf x$. The set of such $\mathbf x$'s can be shown to have cardinality approximately given by $2^{nBc_n}$. Combining this with the above, we have
\begin{align*}
f(\mathbf{y}|\mathbf{i})&=\sum_{\mathbf x} f(\mathbf{y}|\mathbf{x})p(\mathbf{x}|\mathbf{i})\nonumber \\
&\stackrel {.}{\geq}  2^{nBc_n} 2^{-nBb_n}  2^{ -nB \left( \frac{1}{2}\log 2\pi e N +a_n+\sqrt{2a_n\ln2}\log e      \right) }.
\end{align*}
Using the fact  $(\mathbf x, \mathbf i)$ are jointly typical with high probability and given a typical $(\mathbf x, \mathbf i)$ the above lower bound holds for all $\mathbf{y}$ with high probability completes the proof of Lemma \ref{L:Keylemma}. A rigorous proof is given in the sequel.

\appendices

\section{Proof of Lemma \ref{L:Talagrand}}\label{A:ProoftoTalagrand}
Given $A\subseteq \mathbb R^{n}$, let $B:=\{\underline{\omega}\in \mathbb R^{n}  :    \sqrt{N} \underline{\omega} \in A   \}$ and $V_i=\frac{U_i}{\sqrt{N}}, \forall i \in \{1,2,\ldots,n\}$. Then $V_1, V_2, \ldots, V_n$ are $n$ i.i.d. standard Gaussian random variables with $V_i \sim \mathcal{N}(0,1),\forall i \in \{1,2,\ldots,n\}$, and
\begin{align*}
\mbox{Pr}(V^n \in B)= \mbox{Pr}(\sqrt{N} V^n \in A)=\mbox{Pr}(U^n \in A)   \geq 2^{-na_n}.
\end{align*}
We next invoke Gaussian measure concentration as stated in (1.6) of \cite{Talagrand}: for any $B\subseteq \mathbb R^{n}$ and $$t\geq \sqrt{-2\ln \text{Pr}(V^n \in B)},$$ we have
\begin{align*}
\mbox{Pr}(V^n \in \Gamma_{t } (B) )\geq 1-e^{-\frac{1}{2}\left( t- \sqrt{-2\ln \text{Pr}(V^n \in B)}  \right)^2}.
\end{align*}
Thus,  for any $r> 0$,
\begin{align*}
&\mbox{Pr}(V^n \in \Gamma_{\sqrt{n}( \sqrt{2a_n\ln 2 } +\frac{r}{\sqrt{N}} ) } (B) )\\
\geq \ &\mbox{Pr}(V^n \in \Gamma_{\sqrt{-2\ln \text{Pr}(V^n \in B)} + \sqrt{\frac{n}{N}}r    } (B) )\\
\geq \ & 1-2^{-\frac{nr^2}{2N}}.
\end{align*}
Noting that
$$\Gamma_{\sqrt{n}( \sqrt{2Na_n\ln 2 } +r ) } (A)
=\left\{\sqrt{N} \underline{\omega} :     \underline{\omega} \in \Gamma_{\sqrt{n}( \sqrt{2a_n\ln 2} +\frac{r}{\sqrt{N}}) } (B)  \right\},$$
we have
\begin{align*}
&\mbox{Pr}(U^n \in \Gamma_{\sqrt{n}( \sqrt{2Na_n\ln2} +r ) } (A) )\\
=\ &\mbox{Pr}(\sqrt{N}V^n \in \Gamma_{\sqrt{n}( \sqrt{2Na_n\ln2} +r ) } (A) )\\
=\ & \mbox{Pr}(V^n \in \Gamma_{\sqrt{n}( \sqrt{2a_n\ln2} +\frac{r}{\sqrt{N}} ) } (B) )\\
\geq \ & 1-2^{-\frac{nr^2}{2N}}.
\end{align*}

\section{Proof of Lemma \ref{L:Keylemma}}\label{A:ProoftoKeylemma}

\subsection{Definitions of High Probability Sets}
By considering the $B$-length i.i.d. extensions of the $n$-letter random variables involved, law of large numbers allows us to concentrate on a series of ``high probability'' sets defined in the following.\footnote{The high probability sets defined here are analogous to  strongly typical sets that are widely used in information theory. In the Gaussian case, the notion of strong typicality doesn't apply and thus we need to develop our own customized high probability sets. In the discrete memoryless case \cite{WuXieOzgur_ISIT2015}, one can simply resort to strong typicality.}

\vspace{2mm}
\noindent \underline{Definition of $\tilde S(X^n,Z^n)$}
\vspace{2mm}
\begin{lemma}
Assume $H(I_n| X^{n})={n}a_n, H(X^{n}|I_{n})={n}b_n, H(X^{n}|Z^{n})={n}c_n$ for the $n$-channel use code.
Given any $\epsilon >0$ and sufficiently large $B$, we have
\begin{align*}
\mbox{Pr}((\mathbf{X},\mathbf{Z})\in \tilde S(X^n,Z^n))\geq 1-\epsilon
\end{align*}
where
\begin{align*}
 \tilde S(X^n,Z^n):=\Big \{(\mathbf{x},\mathbf{z} ): & ~d(\mathbf x, \mathbf z) \in [\sqrt{nB}(\sqrt{N}-\epsilon),  \sqrt{nB}(\sqrt{N}+\epsilon)    ]\\
 &~2^{-nB(a_n+\epsilon)}\leq p( f(\mathbf{z})|\mathbf{x})\leq 2^{-nB(a_n-\epsilon)} \\
 &~2^{-nB(b_n+\epsilon)}\leq p(\mathbf{x}|f(\mathbf{z}))\leq 2^{-nB(b_n-\epsilon)} \\
  &~2^{-nB(c_n+\epsilon)}\leq p(\mathbf{x}|\mathbf{z})\leq 2^{-nB(c_n-\epsilon)} \Big \}
\end{align*}
\end{lemma}
The lemma is a simple consequence of the law of large numbers.

\vspace{2mm}
\noindent \underline{Definition of $S(X^n,Z^n)$}
\vspace{2mm}

To define $S(X^n,Z^n)$,  we first consider the following lemma, which has been proved in \cite{Zhang}.
\begin{lemma}\label{L:zhanglemma}
Let $A \subseteq C\times D$. For $x\in C$, use $A|_x$ to denote the set $$A|_x=\{y\in D: (x,y)\in A \}.$$ If $\mbox{Pr}(A)\geq 1-\epsilon$,  then
$\mbox{Pr}(B)\geq 1-\sqrt{\epsilon}$,
where $$B:=\{x\in C: \mbox{Pr}(A|_x|x)\geq 1-\sqrt{\epsilon} \}.$$
\end{lemma}

Now, define
\begin{align*}
S(X^n,Z^n)=\{ (\mathbf{x},\mathbf{z})\in \tilde S(X^n,Z^n): \mbox{Pr}(\tilde S(X^n,Z^n)|_\mathbf{z}|\mathbf{z})\geq 1-\sqrt{\epsilon}  \}.
\end{align*}
Clearly $S(X^n,Z^n)$ is a subset of $\tilde S(X^n,Z^n)$. The following lemma says that it is also a high probability set.
\begin{lemma} \label{L: properties_S(x,z)}
$\mbox{Pr}(S(X^n,Z^n))\geq 1- 2\sqrt{\epsilon}$ for $B$ sufficiently large.
\end{lemma}

\begin{proof}
Consider $B$ sufficiently large. Due to Lemma \ref{L:zhanglemma} and the fact that $\mbox{Pr}(\tilde S(X^n,Z^n) )\geq 1-\epsilon$, we have
$$\mbox{Pr}\{ (\mathbf{x},\mathbf{z}): \mbox{Pr}(\tilde S(X^n,Z^n)|_\mathbf{z}|\mathbf{z})\geq 1-\sqrt{\epsilon}  \}\geq 1-\sqrt{\epsilon}.$$
Then by the definition of $S(X^n,Z^n)$,
\begin{align*}
&\mbox{Pr}(S^c(X^n,Z^n))\\
\leq\ & \mbox{Pr}(\tilde S^c(X^n,Z^n))+
\mbox{Pr}\{ (\mathbf{x},\mathbf{z}): \mbox{Pr}(\tilde S(X^n,Z^n)|_\mathbf{z}|\mathbf{z})< 1-\sqrt{\epsilon}  \}\\
\leq \ & \epsilon+\sqrt{\epsilon}\\
\leq \ & 2\sqrt{\epsilon},
\end{align*}
and thus $\mbox{Pr}(S(X^n,Z^n))\geq 1- 2\sqrt{\epsilon}$.
%
%
\end{proof}

\vspace{2mm}
\noindent \underline{Definitions of $\mathcal{Z}_{(\mathbf{x},\mathbf{i})}$ and $S(X^n,I_n)$}
\vspace{2mm}

Define $$\mathcal{Z}_{(\mathbf{x},\mathbf{i})}=\{ \mathbf{z}:  f(\mathbf{z})=\mathbf{i}, (\mathbf{x},\mathbf{z})\in S(X^n,Z^n)   \}$$
and $$S(X^n,I_n)=\{ (\mathbf{x},\mathbf{i}): \mbox{Pr}(\mathcal{Z}_{(\mathbf{x},\mathbf{i})}|\mathbf{x},\mathbf{i})\geq 1-\sqrt[4]{\epsilon} \}.$$

\begin{lemma}\label{L:prob_s(x,i)}
$\mbox{Pr}(S(X^n,I_n))\geq 1-2\sqrt[4]{\epsilon}$ for $B$ sufficiently large.
\end{lemma}
\begin{proof}
For $B$ sufficiently large, consider $\mbox{Pr}( \mathbf{Z} \notin \mathcal{Z}_{(\mathbf{X},\mathbf{I})}     )$.
We have
\begin{align*}
\mbox{Pr}( \mathbf{Z} \notin \mathcal{Z}_{(\mathbf{X},\mathbf{I})}   )=\mbox{Pr}( f(\mathbf{Z})=\mathbf{I},  (\mathbf{X},\mathbf{Z})\notin  S(X^n,Z^n)    )\leq 2 \sqrt{\epsilon}.
\end{align*}

On the other hand,
\begin{align*}
\mbox{Pr}( \mathbf{Z} \notin \mathcal{Z}_{(\mathbf{X},\mathbf{I})}    )
&=\sum_{(\mathbf{x},\mathbf{i})\in S(X^n,I_n)}\mbox{Pr}(\mathbf{Z} \notin \mathcal{Z}_{(\mathbf{x},\mathbf{i})} | \mathbf{x},\mathbf{i}   ) p(\mathbf{x},\mathbf{i})\\
&~~+\sum_{(\mathbf{x},\mathbf{i})\notin S(X^n,I_n)}\mbox{Pr}(\mathbf{Z} \notin \mathcal{Z}_{(\mathbf{x},\mathbf{i})} | \mathbf{x},\mathbf{i}   ) p(\mathbf{x},\mathbf{i})\\
&\geq \sqrt[4]{\epsilon}\cdot \mbox{Pr}( S^c (X^n,I_n)  ).
\end{align*}

Therefore, $\mbox{Pr}( S^c (X^n,I_n)  ) \leq 2\sqrt{\epsilon}/\sqrt[4]{\epsilon}= 2\sqrt[4]{\epsilon}$, and $\mbox{Pr}(S(X^n,I_n))\geq 1-2\sqrt[4]{\epsilon}$.
\end{proof}

\begin{lemma}\label{L: properties_S(x,i)}
For any $(\mathbf{x},\mathbf{i}) \in S(X^n,I_n)$, we have
$$2^{-nB(a_n+\epsilon)}\leq p( \mathbf{i} |\mathbf{x})\leq 2^{-nB(a_n-\epsilon)},$$
and for sufficiently large $B$,
$$\mbox{Pr}(\mathcal{Z}_{(\mathbf{x},\mathbf{i})}|\mathbf{x})\geq 2^{-nB( a_n+2\epsilon )}.$$

\end{lemma}
\begin{proof}
Consider any $(\mathbf{x},\mathbf{i}) \in S(X^n,I_n)$. From the definition of $S(X^n,I_n)$,
$\mbox{Pr}(\mathcal{Z}_{(\mathbf{x},\mathbf{i})}|\mathbf{x},\mathbf{i})\geq 1-\sqrt[4]{\epsilon}$. Therefore,
$\mathcal{Z}_{(\mathbf{x},\mathbf{i})}$ must be nonempty, i.e., there exists at least one $\mathbf{z}\in \mathcal{Z}_{(\mathbf{x},\mathbf{i})}$.

Consider any $\mathbf{z}\in \mathcal{Z}_{(\mathbf{x},\mathbf{i})}$. By the definition of $\mathcal{Z}_{(\mathbf{x},\mathbf{i})}$,
 we have $f(\mathbf{z})=\mathbf{i}$ and $(\mathbf{x},\mathbf{z})\in S(X^n,Z^n)\subseteq \tilde S(X^n,Z^n)$.  Then, it follows from the definition of $\tilde S(X^n,Z^n)$ that
$$2^{-nB(a_n+\epsilon)}\leq p( f(\mathbf{z}) |\mathbf{x})\leq 2^{-nB(a_n-\epsilon)},$$ i.e.,  $$2^{-nB(a_n+\epsilon)}\leq p( \mathbf{i}|\mathbf{x})\leq 2^{-nB(a_n-\epsilon)}.$$
Furthermore,
\begin{align*}
 \mbox{Pr}(\mathbf{Z}\in \mathcal{Z}_{(\mathbf{x},\mathbf{i})}|\mathbf{x})
=\ &\frac{\mbox{Pr}(f(\mathbf{Z})=\mathbf{i}|\mathbf{x})\mbox{Pr}(\mathbf{Z}\in \mathcal{Z}_{(\mathbf{x},\mathbf{i})}|\mathbf{x},f(\mathbf{Z})=\mathbf{i})}
             {\mbox{Pr}(f(\mathbf{Z})=\mathbf{i}|\mathbf{Z}\in \mathcal{Z}_{(\mathbf{x},\mathbf{i})},\mathbf{x})}\\
=\ &p( \mathbf{i}|\mathbf{x}) \mbox{Pr}(\mathcal{Z}_{(\mathbf{x},\mathbf{i})}|\mathbf{x}, \mathbf{i})\\
\geq \ & 2^{-nB(a_n+\epsilon)} (1-\sqrt[4]{\epsilon})\\
\geq\ &2^{-nB(a_n+2\epsilon)}
\end{align*}
for sufficiently large $B$. This finishes the proof of the lemma.
\end{proof}

\subsection{Blowing Up $\mathcal{Z}_{(\mathbf{x},\mathbf{i})}$}

\begin{lemma}\label{L:Blown-Up}
For any $(\mathbf{x},\mathbf{i}) \in S(X^n,I_n)$, consider the following blown-up set of $\mathcal{Z}_{(\mathbf{x},\mathbf{i})}$:
\begin{align*}
&\Gamma_{\sqrt{nB} (\sqrt{2 N a_n}+3\sqrt{N\epsilon}) }(\mathcal{Z}_{(\mathbf{x},\mathbf{i})}) =\{\underline{\mathbf{\omega}}\in \mathbb R^{nB}: \exists \  \underline{\mathbf{\omega}}'\in \mathcal{Z}_{(\mathbf{x},\mathbf{i})} \\
&~~~~~~~~~~~~~~~~~~~~\text{~~s.t.~~} d(\underline{\mathbf{\omega}},\underline{\mathbf{\omega}}')\leq \sqrt{nB} (\sqrt{2 N a_n}+3\sqrt{N\epsilon}) \}.
\end{align*}
We have
\begin{enumerate}
    \item $\mbox{Pr}(\mathbf{Y}\in \Gamma_{\sqrt{nB}  (\sqrt{2Na_n}+3\sqrt{N\epsilon}) }(\mathcal{Z}_{(\mathbf{x},\mathbf{i})}) |\mathbf{x}) \geq 1-\epsilon$ for sufficiently large $B$;
    \item For any $\mathbf{y}\in \Gamma_{\sqrt{nB}  (\sqrt{2Na_n}+3\sqrt{N\epsilon})  }(\mathcal{Z}_{(\mathbf{x},\mathbf{i})}) $,
$$f(\mathbf{y}|\mathbf{i})\geq 2^{-nB(b_n-c_n + \frac{1}{2}\log 2\pi eN +  (a_n+\sqrt{2a_n})\log e   + \epsilon')} $$
where $\epsilon' \to 0$ as $\epsilon \to 0$ and $B\to \infty$.
\end{enumerate}
\end{lemma}

\begin{proof} From Lemma \ref{L: properties_S(x,i)}, for any $(\mathbf{x},\mathbf{i}) \in S(X^n,I_n)$ and sufficiently large $B$,
$$\mbox{Pr}(\mathbf{Z} \in \mathcal{Z}_{(\mathbf{x},\mathbf{i})}|\mathbf{x})\geq 2^{-nB( a_n+2\epsilon )},$$
i.e.,
\begin{align*}
\mbox{Pr}(\mathbf x +\mathbf W_1    \in \mathcal{Z}_{(\mathbf{x},\mathbf{i})}|\mathbf{x})& = \mbox{Pr}( \mathbf W_1    \in    \{ \underline{\omega} -\mathbf x: \underline{\omega} \in  \mathcal{Z}_{(\mathbf{x},\mathbf{i})} \} ) \\
&   \geq 2^{-nB( a_n+2\epsilon )}.
\end{align*}

Therefore, we have
\begin{align}
&\mbox{Pr}(\mathbf{Y}\in \Gamma_{\sqrt{nB}  (\sqrt{2Na_n\ln2}+3\sqrt{N\epsilon}) }(\mathcal{Z}_{(\mathbf{x},\mathbf{i})}) |\mathbf{x}) \nonumber \\
=\ & \mbox{Pr}(\mathbf x +\mathbf W_2    \in   \Gamma_{\sqrt{nB}  (\sqrt{2Na_n\ln2}+3\sqrt{N\epsilon}) }(\mathcal{Z}_{(\mathbf{x},\mathbf{i})}) |\mathbf{x})\nonumber \\
=\ &\mbox{Pr}( \mathbf W_2    \in    \{ \underline{\omega} -\mathbf x: \underline{\omega} \in \Gamma_{\sqrt{nB} (\sqrt{2Na_n\ln2}+3\sqrt{N\epsilon}) }(\mathcal{Z}_{(\mathbf{x},\mathbf{i})}) \}   ) \nonumber \\
=\ &\mbox{Pr}( \mathbf W_1    \in    \{ \underline{\omega} -\mathbf x: \underline{\omega} \in \Gamma_{\sqrt{nB} (\sqrt{2Na_n\ln2}+3\sqrt{N\epsilon}) }(\mathcal{Z}_{(\mathbf{x},\mathbf{i})})) \}   ) \nonumber \\
=\ &\mbox{Pr}( \mathbf W_1    \in   \Gamma_{\sqrt{nB} (\sqrt{2Na_n\ln2}+3\sqrt{N\epsilon}) }(  \{ \underline{\omega} -\mathbf x: \underline{\omega} \in \mathcal{Z}_{(\mathbf{x},\mathbf{i})}\} ) ) \nonumber \\
\geq \ &\mbox{Pr}( \mathbf W_1    \in   \Gamma_{\sqrt{nB} (\sqrt{2Na_n\ln2+ 4N\epsilon\ln2}  +\sqrt{N\epsilon}      ) }(  \{ \underline{\omega} -\mathbf x: \underline{\omega} \in \mathcal{Z}_{(\mathbf{x},\mathbf{i})} \} ) ) \nonumber \\
\geq \ & 1-2^{-\frac{nB\epsilon}{2}} \label{E:followfromtalagrand}\\
\geq \ &1-\epsilon \nonumber
\end{align}
for sufficiently large $B$, where \dref{E:followfromtalagrand} follows from Lemma \ref{L:Talagrand}.

To prove Part 2), consider any $\mathbf{y}\in \Gamma_{\sqrt{nB}  (\sqrt{2Na_n\ln2}+3\sqrt{N\epsilon}) }(\mathcal{Z}_{(\mathbf{x},\mathbf{i})})$. We can find one
$\mathbf{z}\in  \mathcal{Z}_{(\mathbf{x},\mathbf{i})}$ such that $d (\mathbf{y},\mathbf{z})\leq \sqrt{nB}  (\sqrt{2Na_n\ln2}+3\sqrt{N\epsilon})$, and for this
$\mathbf{z}$, we have from the definition of $\mathcal{Z}_{(\mathbf{x},\mathbf{i})}$ that: i) $f(\mathbf{z})=\mathbf i$ and ii) $\mbox{Pr}(\tilde S(X^n,Z^n)|_\mathbf{z}|\mathbf{z})\geq 1-\sqrt{\epsilon}$, where
\begin{align*}
\tilde S(X^n,Z^n)|_\mathbf{z}=\Big\{ \mathbf{x}: & ~d(\mathbf x, \mathbf z) \in [\sqrt{nB}(\sqrt{N}-\epsilon),  \sqrt{nB}(\sqrt{N}+\epsilon)    ]\\
 &~2^{-nB(a_n+\epsilon)}\leq p( f(\mathbf{z})|\mathbf{x})\leq 2^{-nB(a_n-\epsilon)}\\
 &~2^{-nB(b_n+\epsilon)}\leq p(\mathbf{x}|f(\mathbf{z}))\leq 2^{-nB(b_n-\epsilon)}  \\
  &~2^{-nB(c_n+\epsilon)}\leq p(\mathbf{x}|\mathbf{z})\leq 2^{-nB(c_n-\epsilon)} \Big\}.
\end{align*}
The size of $\tilde S(X^n,Z^n)|_\mathbf{z}$ can be lower bounded by considering the following
\begin{align*}
1-\sqrt{\epsilon} &\leq \mbox{Pr}(\tilde S(X^n,Z^n)|_\mathbf{z}|\mathbf{z})\\
&= \sum_{\mathbf x \in \tilde S(X^n,Z^n)|_\mathbf{z}} p(\mathbf x|\mathbf z)\\
&\leq 2^{-nB(c_n-\epsilon)} \big|\tilde S(X^n,Z^n)|_\mathbf{z}\big|,
\end{align*}
i.e.,
\begin{align*}\big|\tilde S(X^n,Z^n)|_\mathbf{z}\big| \geq (1-\sqrt{\epsilon} )2^{nB(c_n-\epsilon)}.
\end{align*}

Then,
\begin{align}
f(\mathbf{y}|\mathbf{i})&=\sum_{\mathbf x} f(\mathbf{y}|\mathbf{x})p(\mathbf{x}|\mathbf{i})\nonumber \\
&\geq \sum_{\mathbf x \in \tilde S(X^n,Z^n)|_\mathbf{z}} f(\mathbf{y}|\mathbf{x})p(\mathbf{x}|\mathbf{i})\nonumber \\
&\geq 2^{-nB(b_n+\epsilon)} \sum_{\mathbf x \in \tilde S(X^n,Z^n)|_\mathbf{z}} f(\mathbf{y}|\mathbf{x})\nonumber \\
&\geq 2^{-nB(b_n+\epsilon)} \big|\tilde S(X^n,Z^n)|_\mathbf{z}\big| \min_{\mathbf x \in \tilde S(X^n,Z^n)|_\mathbf{z}} f(\mathbf{y}|\mathbf{x}) \nonumber \\
&\geq  (1-\sqrt{\epsilon} )2^{-nB(b_n+\epsilon)}2^{nB(c_n-\epsilon)} \min_{\mathbf x \in \tilde S(X^n,Z^n)|_\mathbf{z}} f(\mathbf{y}|\mathbf{x}). \label{E:probtobecont}
\end{align}
For any $\mathbf x \in \tilde S(X^n,Z^n)|_\mathbf{z}$, we have
\begin{align*}
d(\mathbf x, \mathbf y)&\leq d(\mathbf x, \mathbf z)+d(\mathbf z, \mathbf y)\\
&\leq \sqrt{nB}(\sqrt{N}+\sqrt{2Na_n\ln2}+\epsilon+3\sqrt{N\epsilon})\\
&=: \sqrt{nB}(\sqrt{N}+\sqrt{2Na_n\ln2}+\epsilon_1)
\end{align*}
and thus,
\begin{align*}
f(\mathbf y| \mathbf x)&=  \frac{1}{ (2\pi N)^{ \frac{nB}{2} } } e^{ -\frac{|| \mathbf y- \mathbf x  ||^2     }{2N}  } \\
&\geq    2^{ -\frac{     nB (\sqrt{N}+\sqrt{2Na_n\ln2}+\epsilon_1)    ^2     }{2N}  \log e  -\frac{nB}{2}\log 2\pi N   } \\
&=    2^{ -nB \left(   \frac{      (\sqrt{N}+\sqrt{2Na_n\ln2}+\epsilon_1)    ^2     }{2N}  \log e + \frac{1}{2}\log 2\pi N  \right) } \\
&=:    2^{ -nB \left( \frac{1}{2}\log 2\pi e N +a_n+\sqrt{2a_n\ln2}\log e   +\epsilon_2  \right) }
\end{align*}
where $\epsilon_1, \epsilon_2 \to 0$ as $\epsilon \to 0$. Plugging this into \dref{E:probtobecont} yields that
\begin{align*}
f(\mathbf y|\mathbf i)&\geq(1-\sqrt{\epsilon} )2^{-nB(b_n+\epsilon)}2^{nB(c_n-\epsilon)}\\
&~~~~\times 2^{ -nB \left( \frac{1}{2}\log 2\pi e N +a_n+\sqrt{2a_n\ln2}\log e    +\epsilon_2  \right) }  \\
&\geq 2^{-nB(b_n-c_n + \frac{1}{2}\log 2\pi e N+a_n+\sqrt{2a_n\ln2}\log e  + \epsilon_3)}
\end{align*}
for some $\epsilon_3 \to 0$ as $\epsilon \to 0$.
\end{proof}

\subsection{Constructions of $\mathcal I$ and $\mathcal{Y}_\mathbf{i}$}
Let $\mathcal I=\{\mathbf{i}: \mbox{Pr}(S(X^n,I_n)|_\mathbf{i}|\mathbf{i}) \geq 1-2\sqrt[8]{\epsilon} \}$.  For sufficiently large $B$, $\mbox{Pr}(S(X^n,I_n))\geq 1-2\sqrt[4]{\epsilon}$
from Lemma \ref{L:prob_s(x,i)}, and thus by Lemma \ref{L:zhanglemma} again,
\begin{align*}
\mbox{Pr}(\mathcal I)&\geq  \mbox{Pr}\left\{\mathbf{i}: \mbox{Pr}(S(X^n,I_n)|_\mathbf{i}|\mathbf{i}) \geq 1-   \sqrt{ 2\sqrt[4]{\epsilon}} \right\}\\
&\geq 1-   \sqrt{ 2\sqrt[4]{\epsilon}} \\
&\geq 1-2\sqrt[8]{\epsilon}.
\end{align*}
\begin{lemma}\label{L:AppendixY_i}
For any $\mathbf{i} \in \mathcal I$, let
\begin{align*}
 \mathcal{Y}_\mathbf{i}&:=\bigcup_{\mathbf{x} \in S(X^n,I_n)|_\mathbf{i}} \Gamma_{\sqrt{nB} (\sqrt{2Na_n\ln2}+3\sqrt{N\epsilon})  }(\mathcal{Z}_{(\mathbf{x},\mathbf{i})}).
 \end{align*}
Then for sufficiently large $B$,
$$\mbox{Pr}(\mathbf{Y} \in \mathcal{Y}_\mathbf{i} |\mathbf{i})\geq 1-3\sqrt[8]{\epsilon},$$
and for each $\mathbf y \in \mathcal{Y}_\mathbf{i}$,
\begin{align*}
f(\mathbf{y}|\mathbf{i})\geq 2^{-nB(b_n-c_n + \frac{1}{2}\log 2\pi eN +  a_n+\sqrt{2a_n\ln2}\log e    + \epsilon_3)}. \end{align*}
\end{lemma}
\begin{proof}
For any $\mathbf{i} \in \mathcal I$ and sufficiently large $B$, we have
\begin{align*}
&\mbox{Pr}(\mathbf{Y} \in \mathcal{Y}_\mathbf{i} |\mathbf{i})\\
=\ &\sum_{\mathbf{x}}\mbox{Pr}(\mathbf{Y} \in \mathcal{Y}_\mathbf{i} |\mathbf{x}) p(\mathbf{x}|\mathbf{i})\\
\geq \ &\sum_{\mathbf{x}\in S(X^n,I_n)|_\mathbf{i}}\mbox{Pr}(\mathbf{Y} \in \mathcal{Y}_\mathbf{i} |\mathbf{x}) p(\mathbf{x}|\mathbf{i})\\
\geq \ & \sum_{\mathbf{x}\in S(X^n,I_n)|_\mathbf{i}}\mbox{Pr}(\mathbf{Y} \in \Gamma_{\sqrt{nB} (\sqrt{2Na_n\ln2}+3\sqrt{N\epsilon})  }(\mathcal{Z}_{(\mathbf{x},\mathbf{i})}) |\mathbf{x})  p(\mathbf{x}|\mathbf{i})\\
\geq \ & (1-\epsilon) \mbox{Pr}(S(X^n,I_n)|_\mathbf{i}   |\mathbf i)\\
\geq \ &(1-\epsilon) (1-2\sqrt[8]{\epsilon})\\
\geq \ &1-3\sqrt[8]{\epsilon}.
\end{align*}

Now consider any $\mathbf y \in \mathcal{Y}_\mathbf{i}$. There exists some $\mathbf{x} \in S(X^n,I_n)|_\mathbf{i}$ such that
$\mathbf y \in \Gamma_{\sqrt{nB} (\sqrt{2Na_n\ln2}+3\sqrt{N\epsilon})  }(\mathcal{Z}_{(\mathbf{x},\mathbf{i})})$. It then follows immediately from Part 2) of Lemma \ref{L:Blown-Up} that
\begin{align*}
 f(\mathbf{y}|\mathbf{i})\geq 2^{-nB(b_n-c_n + \frac{1}{2}\log 2\pi eN + a_n+\sqrt{2a_n\ln2}\log e    + \epsilon_3)} . \end{align*}
\end{proof}

Finally, choosing $\delta$ to be $3\sqrt[8]{\epsilon}$ completes the proof of Lemma \ref{L:Keylemma}.


\begin{thebibliography}{10}
\bibitem{van71}
E.~C. {van der Meulen}, ``Three-terminal communication channels,'' {\em Adv.
  Appl. Prob.}, vol.~3, pp.~120--154, 1971.

\bibitem{covelg79}
T.~Cover and A.~{El Gamal}, ``Capacity theorems for the relay channel,'' {\em
  {IEEE} Trans. Inform. Theory}, vol.~25, pp.~572--584, 1979.

\bibitem{kramer} G. Kramer, M. Gastpar, and P. Gupta, ``Cooperative Strategies and Capacity Theorems for Relay Networks,'' \emph{IEEE Trans. Info. Theory}, vol. 51, no. 9, pp. 3037--3063, Sept. 2005.

\bibitem{schein}
B. Schein and R. Gallager, ``The Gaussian parallel relay network,'' in \emph{Proc. of IEEE International Symposium on Information Theory}, pp. 22, June 2000.

\bibitem{KimAllerton}
Y.-H. Kim, ``Coding techniques for primitive relay channels,'' in \emph{Proc. Forty-Fifth Annual Allerton Conf. Commun., Contr. Comput.}, Monticello, IL, Sep. 2007.

\bibitem{Avestimehretal} A. S. Avestimehr, S. N. Diggavi, and D. N. C. Tse, ``Wireless Network Information Flow: A Deterministic Approach,'' \emph{IEEE Trans. Info. Theory}, vol.~57, no.~4, pp.~1872--1905, 2011.

\bibitem{bobak} B. Nazer and M. Gastpar, ``Compute-and-forward: Harnessing interference through structured codes,'' \emph{IEEE Trans. Inf. Theory}, vol. 57, no. 10, pp. 6463--6486, 2011.


\bibitem{OzgurDiggavi}
A. Ozgur and S N. Diggavi, ``Approximately achieving Gaussian relay network capacity with
lattice-based QMF codes,'' \emph{IEEE Trans. Info. Theory}, vol. 59, no. 12, pp. 8275--8294, December 2013.

\bibitem{Limetal} S. H. Lim, Y.-H. Kim, A. El Gamal, S.-Y. Chung, ``Noisy network coding,'' \emph{IEEE Trans. Info. Theory}, vol. 57, no. 5, pp. 3132--3152, May 2011.




\bibitem{Koetteretal} R. Koetter, M. Effros, and M. M\'{e}dard,  ``A
theory of network equivalence---Part I: Point-to-Point Channels,'' \emph{IEEE Trans. Info. Theory}, vol. 57, no. 2, pp. 972--995, February 2011.


\bibitem{CourtadeOzgur} T. Courtade and A Ozgur, ``Approximate capacity of Gaussian relay networks: Is a sublinear gap to the cutset bound plausible?'' in \emph{Proc. of IEEE International Symposium on Information Theory}, Hong Kong, June 2015.


%
%

\bibitem{Zhang}
Z. Zhang,  ``Partial converse for a relay channel,''  {\em
{IEEE} Trans. Inform. Theory}, vol.~34, no. 5,  pp.~1106--1110, Sept. 1988.


\bibitem{Xue}
F. Xue,  ``A new upper bound on the capacity of a primitive relay channel based on channel simulation,''  {\em
{IEEE} Trans. Inform. Theory}, vol.~60,  pp.~4786--4798, Aug. 2014.


%
%
%
%
%



\bibitem{WuXieOzgur_ISIT2015}
X. Wu, L.-L. Xie, A. Ozgur, ``Upper bounds on the capacity of symmetric primitive relay channels,''
in \emph{Proc. of IEEE International Symposium on Information Theory}, Hong Kong, June 2015.






\bibitem{Urs1}
 U. Niesen, B. Nazer, and P. Whiting, ``Computation alignment: Capacity
approximation without noise accumulation,'' {\em
{IEEE} Trans. Inform. Theory}, vol. 59, no. 6, pp. 3811--3832, 2013.

\bibitem{Urs2}
 U. Niesen and S. Diggavi, ``The approximate capacity of the Gaussian
n-relay diamond network,'' {\em {IEEE} Trans. Inform. Theory}, vol. 59, no. 2, pp. 845--859, Feb 2013.

\bibitem{Bobbie}
B. Chern and A. Ozgur, ``Achieving the capacity of the n-relay Gaussian
diamond network within log n bits,'' in
\emph{Proc. of IEEE Information Theory Workshop},  2012.

\bibitem{Ritesh}
R. Kolte, and A. Ozgur, ``Improved capacity approximations for Gaussian relay networks,'' in
\emph{Proc. of IEEE Information Theory Workshop}, 2013.

\bibitem{DDF}
S. H. Lim, K. T. Kim, and Y.-H. Kim, ``Distributed decode-forward for multicast,'' in \emph{Proc. of IEEE International Symposium on Information Theory}, pp. 636--640, July 2014.

\bibitem{Talagrand}
M. Talagrand, ``Transportation cost for Gaussian and other product measures,'' {\em Geometric \& Functional Analysis}, pp. 587--600.

%
%
%
%

\end{thebibliography}
\end{document}